\documentclass[letterpaper, 10 pt, conference]{ieeeconf}  

\IEEEoverridecommandlockouts                              

\usepackage{graphics}
\usepackage{float}
\usepackage{epsfig} 
\usepackage{mathptmx} 
\usepackage{times} 
\usepackage{amsmath} 
\usepackage{amssymb}  
\usepackage{braket}
\usepackage{dsfont}
\newtheorem{theorem}{Theorem}[section]
\newtheorem{definition}{Definition}[section]
\newtheorem{assumption}{Assumption}[section]

\newcommand{\tr}{\operatorname{tr}}

\renewcommand{\H}{\mathcal H}

\renewcommand{\P}{\mathbb P}

\renewcommand{\o}{\left(}
\renewcommand{\c}{\right)}

\newcommand{\Y}{\mathcal Y}

\newcommand{\thetahat}{\hat \theta}

\renewcommand{\L}{\mathcal L (\mathcal H) }
\newcommand{\D}{\mathcal D (\mathcal H) }
\usepackage{multicol, blindtext}

\title{\LARGE \bf
On Asymptotic Stability of  Non-Demolition Quantum Trajectories with Measurement Imperfections
}

\author{Maël Bompais$^{1}$ and Nina Amini$^{1}$ 
\thanks{$^{1}$Ma\"el Bompais and Nina H. Amini are with Laboratoire des Signaux et Syst\`{e}mes (L2S), CNRS-CentraleSup\'{e}lec-Universit\'{e} Paris-Sud, Universit\'{e} Paris-Saclay, 3, Rue Joliot Curie, 91190 Gif-sur-Yvette, France
  ({first name.last name@centralesupelec.fr}.)}}
\begin{document}

\maketitle
\thispagestyle{empty}
\pagestyle{empty}

\begin{abstract}
We consider the question of asymptotic stability of quantum trajectories undergoing quantum non-demolition imperfect measurement, that is to say the convergence of the estimated trajectory towards the true trajectory whose parameters and initial state are not necessarily known. We give conditions on the estimated initial state and regions of validity for the estimated parameters so that this convergence is ensured. We illustrate these results through numerical simulations on the physical example \cite{sayrin2011real} and discuss the asymptotic stability for a more  realistic general case where decoherence acts on the system. In this case, the evolution is described by new Kraus operators which do not satisfy the quantum non-demolition property.
\end{abstract}
\section{INTRODUCTION}
Open quantum systems are systems which are in interaction with an environment. Such an interaction causes the loss of information from the system to the environment or vice versa from the environment to the system, a phenomenon called decoherence, see e.g., \cite{breuer2002theory}. Open quantum systems can be observed through measurement processes, see e.g., \cite{barchielli2009quantum,wiseman2009quantum}. 

Direct measurement of a quantum systems will freeze its state, a phenomenon called Zeno effect (see e.g., \cite{barchielli2009quantum}). Instead, the quantum system can be measured indirectly. This means that the system becomes in interaction with a meter, usually light beams, then an observable of the scattered light is measured. This measure contains information about the system and also has random back-action on it. Quantum measurements have completely probabilistic nature. In discrete-time, the evolution is described by a Markov chain. In continuous-time, it is described by stochastic master equations driven by Wiener and Poisson processes depending on the type of detection; Wiener processes for homodyne or heterodyne and Poisson processes for photon counting detection, see e.g., \cite{walls2008bose}. Such stochastic processes are referred to as quantum trajectories (see \cite{carmichael2009open}) or quantum filters (see \cite{belavkin1995quantum,belavkin1992quantum,bouten2007introduction}). 

The perturbation induced by the measurement on the system can be better overcome thanks to the quantum non-demolition (QND) property which is introduced in \cite{braginsky1995quantum}. In essence, this condition means that it is possible to identify a basis for the system Hilbert space where every element of the basis remains unchanged by measurement. The QND measurement has been implemented by Serge Haroche's group to realize a first experiment of a feedback stabilizing photon number states inside a quantum electrodynamics cavity \cite{haroche2013nobel,raimond2006exploring,sayrin2011real}.


In real experiments, there are many sources of imperfection related for instance to the detector efficiency, the error in the results, etc. Quantum trajectories in presence of imperfections take more complicated forms. 
Robustness of quantum trajectories with respect to misspecification of parameters and initial states is fundamental. For instance, error correction codes and quantum feedback control play an important role to keep the behavior of trajectories robust to errors, see e.g., \cite{lidar2013quantum,wiseman1994quantuma}. 

The asymptotic stability of quantum trajectories, i.e., the convergence of the estimated trajectories towards the true ones, in the case of lack of knowledge about initial state, has been investigated in different papers with different approaches, whereas developing necessary conditions for generic quantum trajectories including measurement imperfections remains open.
In \cite{van2009stability}, the author establishes a sufficient condition for asymptotic stability. In \cite{bauer2013repeated,benoist2014large}, for the case of perfect QND measurement, the authors obtain asymptotic stability for discrete-time and continuous time trajectories. In \cite{amini2021asymptotic}, we characterize necessary conditions for asymptotic stability of quantum trajectories for perfect measurements. In \cite{somaraju2012design}, the authors design an optimal discrete-time filter containing different sources of imperfections and show that such the filter in average becomes closer to the ideal filter with correct initial state by proving the sub-martingale property of the fidelity (see \cite[Chapter 9]{nielsen2010quantum} for the definition of the fidelity). Later in \cite{amini2014stability}, this result is extended to the continuous-time filter.



Here we consider discrete-time quantum trajectories when measurements are QND and imperfect, more precisely we consider the optimal filter developed in \cite{somaraju2012design}. Firstly, we study a quantum state reduction property for these quantum trajectories where we obtain the rate of convergence. To this end, we adopt the approaches developed in \cite{bauer2013repeated} for the case of perfect QND measurements (see Section \ref{sec:initial}). Secondly, we assume that the initial state and the physical parameters appearing in the measurement operators are unknown, and we develop asymptotic stability analyses in these cases where we provide conditions on the estimated parameters and estimated initial state which ensure the asymptotic stability (see Section \ref{sec:asy}). This represents the first result on such kind of asymptotic stability including misspecification of parameters in addition to initial state. Finally, we consider the well-known example of the photon box \cite{sayrin2011real}, where we observe numerically such the asymptotic stability. Furthermore, we consider the case of decoherence induced by the environment in this example, and we illustrate numerically that even though new Kraus operators do not satisfy the QND condition, the asymptotic stability property with respect to initial state is still ensured (see Section \ref{sec:numeric}). This is encouraging for further investigations in this direction. 
\section{Selection of the pointer state for QND imperfect measurements}
\label{sec:initial}
\subsection{Quantum trajectories with measurement imperfections}
We consider a quantum system undergoing  discrete-time indirect measurement. The Hilbert space of the system is denoted by $\H \cong\mathbb  C^d$ and the state space is the set of density matrices
$$\D=\{ \rho \in \L |\, \rho=\rho^\dag,\,\rho \geq 0,\, \tr(\rho)=1\}.$$
With an ideal detector, the state of the system evolving over time satisfies the following Markov chain
$${\rho}_{n+1}=\frac{V_{i_n}{\rho}_{n}V_{i_n}^\dag}{\operatorname{tr}\left(V_{i_n}{\rho}_{n}V_{i_n}^\dag\right)},$$
depending on the random measurement result $i_n$ taking values in a finite set $\Y.$ The probability to detect the measurement result $i_n=i$ at time $n$ is given by $\tr\left(V_{i}{\rho}_{n}V_{i}^\dag\right).$ The operators $V_i$ are referred to as Kraus operators and satisfy $\sum_{i \in \Y} V_i^\dag V_i=I,$ with $I$ denoting the identity operator on $\H.$\\

Now suppose that the detector is biased and give some corrupted measurement results. The probability of error is supposed to be known and described by the correlation matrix $\eta$ such that $\eta_{i,j}$ is the probability to detect $i$ while an ideal detector would have given $j.$ In this case, the best estimation of the state of the system knowing the sequence of measurement results is given by the optimal filter derived in \cite{somaraju2012design}
\begin{equation}{\rho}_{n+1}=\frac{\boldsymbol{\Phi}_{i_{n}}\left({\rho}_{n}\right)}{\operatorname{tr}\left(\boldsymbol{\Phi}_{i_{n}}\left({\rho}_{n}\right)\right)}
\label{eq:dynamic}
\end{equation}
\noindent with $\boldsymbol{\Phi}_{i}(\rho)=\sum_{i\in \Y} \eta_{i,j} V_j \rho V_j^\dag$ and the measurement results are governed by the probability measure $\P$ such that $\P(i_n=i)=\operatorname{Tr}\left(\boldsymbol{\Phi}_{i_{n}}\left({\rho}_{n}\right)\right).$\\

\subsection{Asymptotic behaviour under QND imperfect measurement}
In this section, we aim to study the asymptotic behavior of the Markov chain \eqref{eq:dynamic} when the measurement satisfies a QND condition defined as follows:
\vspace{0.2cm}
\begin{definition}[QND measurement]
A measurement satisfies a non-demolition condition for a basis $\mathcal P$ if any element $\ket\alpha$ of $\mathcal P$ is not changed by the measurement. The states $\ket\alpha\bra\alpha$  are called pointer states.
\end{definition}
\vspace{0.2cm}
\noindent We consider from now that this condition is verified. For $\ket\alpha$ an element in $\mathcal P,$ we define
$$p(i|\alpha)=\tr \o \Phi_i(\ket{\alpha}\bra{\alpha}) \c$$
which correspond to the probability to observe $i$ when we perform a measure on the state $\ket{\alpha}\bra{\alpha}.$ Let us also define $$p_n(i)=\tr \o \boldsymbol{\Phi}_{i}(\rho_n) \c$$
which corresponds to the probability to observe $i$ at time $n.$

Throughout this paper, we suppose that the following non-degeneracy assumption holds true. 
\medskip
\begin{assumption}
For $\alpha \neq \beta \in \mathcal P$, there exists $i \in \Y$ such that:
$$  p(i|\alpha) \neq  p(i|\beta).$$
\label{ND}
\end{assumption}
In order to study the asymptotic behaviour of $\rho_n$, we will be interested in the quantities
$$ q_{\alpha}(n)=\tr (\ket{\alpha}\bra{\alpha} \rho_n),$$
which determines the population in the pointer state $\ket\alpha\bra\alpha$ at time $n.$ Note that $q_\alpha(n) \in [0,1]$ and $\sum_{\beta \in \mathcal P}q_\beta (n)=1.$\\
The following theorem establishes the selection of a pointer state.
\medskip
\begin{theorem}
There exists a random variable $\Upsilon$ taking values in the set $\mathcal P$ such that:
$$\lim\limits_{n \to \infty} q_{\Upsilon}(n)=1 ~~a.s.,$$
equivalently
$$\lim\limits_{n \to \infty} \rho_n = \ket{\Upsilon}\bra{\Upsilon} ~~a.s.$$
Moreover, $\P(\Upsilon=\alpha)=q_{\alpha}(0).$ 
\end{theorem}
\medskip
The proof is a direct adaptation of the arguments applied in \cite{bauer2011convergence} for perfect measurement. This uses the martingale property of $q_\alpha(n),$ still valid for imperfect measurements, and Assumption \ref{ND}, to state that $q_\alpha(\infty)$ and $q_\beta(\infty)$ cannot be non-zero simultaneously for $\alpha \neq \beta.$

\medskip

Now let us set the following notation $\P_\alpha (.)=p(.|\alpha).$ 
The next theorem precises the speed of selection of the pointer state $\Upsilon$.
\medskip
\begin{theorem}
Let $\alpha$ be such that $q_{\alpha}(0) \neq 0.$ Then:
$$\lim\limits_{n \to \infty} \frac{1}{n}\ln \o \frac{q_{\alpha}(n)}{ q_{\Upsilon}(n)} \c = - S(\P_{\Upsilon} || \P_{\alpha}) ~~\textrm{a.s.,} $$
where $S(\P_{\Upsilon} || \P_{\alpha})=\sum_i  p(i|\Upsilon)\ln(\frac{ p(i|\Upsilon)}{ p(i|\alpha)})$ is the relative entropy between the probability distributions $\P_{\Upsilon}$ and $\P_{\alpha}$.
\end{theorem}
\medskip
\begin{proof}
First we start by showing that the following recurrence 
\begin{equation}
q_{\alpha}(n+1)=q_{\alpha}(n)\frac{ p(i_n|\alpha)}{ p_n(i_n)} 
\label{eq:recurrence}
\end{equation}
holds for $q_{\alpha}(n).$
We note that 
\begin{align*}
q_{\alpha}(n+1)&=\tr\left(\ket\alpha\bra\alpha\rho_{n+1}\right)\\
&=\tr\left(\ket\alpha\bra\alpha\frac{\boldsymbol\Phi_{i_n}(\rho_n)}{\tr\left(\boldsymbol\Phi_{i_n}(\rho_n)\right)}\right)\\
&=\frac{\tr\left(\boldsymbol\Phi_{i_n}^\dag(\ket\alpha\bra\alpha)\rho_n\right)}{\tr\left(\boldsymbol\Phi_{i_n}(\rho_n)\right)}\\
&=\frac{\tr\left(\tr(\boldsymbol\Phi_{i_n}^\dag(\ket\alpha\bra\alpha))\ket\alpha\bra\alpha\rho_n\right)}{\tr\left(\boldsymbol\Phi_{i_n}(\rho_n)\right)},
\end{align*}
where for the last equality, we use the QND property of measurement. Then we get the following 
\begin{align*}
q_{\alpha}(n+1)=\frac{\tr(\boldsymbol\Phi_{i_n}^\dag(\ket\alpha\bra\alpha))\tr\left(\ket\alpha\bra\alpha\rho_n\right)}{\tr\left(\boldsymbol\Phi_{i_n}(\rho_n)\right)},
\end{align*}
which gives the desired recurrence relation \eqref{eq:recurrence}. The rest of the proof uses similar  arguments as in  \cite[Section 4.3]{bauer2013repeated}. For the sake of readability, we give the main parts of these arguments. 

First we note that the relation \eqref{eq:recurrence} can be rewritten as follows
\begin{equation}
q_{\alpha}(n+1)=q_{\alpha}(0)\frac{\prod_{k=1}^n p(i_k|\alpha)}{\sum_{\beta}q_\beta(0)\prod_{k=1}^n p(i_k|\beta)}.
\end{equation}
Now we can evaluate the following ratio for $\alpha$ and $\zeta$ in $\mathcal P$
\begin{align*}
&\frac{1}{n}\ln \o \frac{q_{\alpha}(n)}{ q_{\zeta}(n)}\c\\
&=\frac 1 n\ln \o \frac{q_{\alpha}(0)}{ q_{\zeta}(0)}\c+\frac{1}{n}\sum_{k=1}^n\ln( p(i_k|\alpha))-\ln( p(i_k|\zeta)).
\end{align*}
Then it is sufficient to use the fact that the results of measurements are identically independently distributed under the new probability $\mathbb Q_\gamma$ defined by $\mathbb Q_\gamma(.)=\P(.|\Upsilon=\gamma),$ we can conclude the following relation by the law of large number 
\begin{equation}
\mathbb Q_\zeta\o\lim\limits_{n \to \infty} \frac{1}{n}\ln \o \frac{q_{\alpha}(n)}{ q_{\zeta}(n)}\c =- S(\P_{\zeta} || \P_{\alpha})\c=1.\,\quad
\label{eq:prob}
\end{equation}
However we aim to show the following 
$$\P\o \lim\limits_{n \to \infty} \frac{1}{n}\ln \o \frac{q_{\alpha}(n)}{ q_{\Upsilon}(n)} \c=- S(\P_{\Upsilon} || \P_{\alpha})\c=1.$$
Now we remark that $$\P(.)=\sum_\gamma \P(\Upsilon=\gamma)\P(.|\Upsilon=\gamma)=\sum_\gamma q_\gamma(0)\mathbb Q_\gamma(.)$$ Then we have 
\begin{align*}
&\P\o \lim\limits_{n \to \infty} \frac{1}{n}\ln \o \frac{q_{\alpha}(n)}{q_{\Upsilon}(n)} \c=- S(\P_{\Upsilon} || \P_{\alpha})\c\\
&=\sum_{\gamma}q_\gamma(0)\mathbb Q_\gamma\o \lim\limits_{n \to \infty} \frac{1}{n}\ln \o \frac{q_{\alpha}(n)}{ q_{\gamma}(n)} \c=- S(\P_{\gamma} || \P_{\alpha})\c\\
&=\sum_{\gamma}q_\gamma(0)=1.
\end{align*}
\end{proof}
\section{Asymptotic stability}
\label{sec:asy}
In this section, we show asymptotic stability of the Markov chain \eqref{eq:dynamic} for two cases, first in absence of knowledge about initial state and second for ignorance of both initial state and physical parameters.
\subsection{Unknown initial state}
When the initial state $\rho_0$ is unknown, a natural way to construct an estimation $\hat\rho_n$ of the true trajectory $\rho_n$ is to fix an arbitrary estimated initial state $\hat\rho_0$ and make it evolve according to the measurement results $i_n$ at our disposal (emitted by the true system $\rho_n$), with the recursive relation
\begin{equation}{\hat\rho}_{n+1}=\frac{\boldsymbol{\Phi}_{i_{n}}\left({\hat\rho}_{n}\right)}{\operatorname{Tr}\left(\boldsymbol{\Phi}_{i_{n}}\left({\hat\rho}_{n}\right)\right)}.
\label{eq:dynamicest}
\end{equation}
Similarly to the previous study, we define $ \hat q_{\alpha}(n)=\tr (\ket{\alpha}\bra{\alpha} \hat\rho_n).$ The following theorem shows the asymptotic stability with respect to the initial state.
\medskip
\begin{theorem}
Let $ \alpha $ be such that $\hat q_{\alpha}(0) \neq 0$. Then
$$\lim\limits_{n \to \infty} \frac{1}{n}\ln \o \frac{\hat q_{\alpha}(n)}{\hat q_{\Upsilon}(n)} \c = - S(\P_{\Upsilon} || \P_{\alpha})  $$
almost surely, where $\Upsilon$ is the random variable designating the pointer state selected by the trajectory $\rho_n$. Consequently, if for all $\alpha$, $q_\alpha(0)\neq 0$, then
$$ \lim\limits_{n \to \infty} \hat q_{\Upsilon}(n)=1,$$
equivalently
$$\lim\limits_{n \to \infty} \hat\rho_n =\ket{\Upsilon}\bra{\Upsilon}~a.s.$$

\end{theorem}

\vspace{0.3cm}
 In other words, the estimated trajectory $\hat\rho_n$ selects the same pointer state as the true trajectory $\rho_n$.
 A simple way to fulfill the condition $q_\alpha(0)\neq 0$ for all $\alpha$ is to choose $\hat\rho_0$ as a full rank state.
 
\medskip
 \begin{proof}
 The proof is principally based on showing the same recurrence property announced in \eqref{eq:recurrence} for estimated trajectory \eqref{eq:dynamicest}, i.e., for $\hat q_{\alpha}(n),$ which can be obtained in the same manner as before. The rest of the proof is similar to \cite[Section 4.3]{bauer2013repeated}.   
 \end{proof}
\subsection{Unknown initial state $\&$ unknown parameters}
Suppose now that the Kraus maps $\boldsymbol\Phi_i$ depend on an unknown parameter $\theta$, we write them $\boldsymbol\Phi_i^\theta$, we then have a new estimation of the trajectory, which is based on an estimation $\hat \theta$ of $\theta$; and again evolving according to the measurement results detected as
$${\hat\rho}_{n+1}^{\thetahat}=\frac{\boldsymbol{\Phi}_{i_{n}}^{\thetahat}\left({\hat\rho}_{n}^{\thetahat}\right)}{\operatorname{Tr}\left(\boldsymbol{\Phi}_{i_{n}}^{\thetahat}\left({\hat\rho}_{n}^{\thetahat}\right)\right)}.$$
We define the quantities $\hat q_{\alpha}^{\hat\theta}(n)=\tr (\ket{\alpha}\bra{\alpha} \hat\rho_n^{\hat\theta}).$ We also set the notation $\P_\alpha^{\hat\theta} (.)=p^{\hat\theta}(.|\alpha)$ where $p^{\hat\theta}(i|\alpha)=\tr \o \Phi_i^{\hat\theta}(\ket{\alpha}\bra{\alpha}) \c$. The following theorem establishes the asymptotic stability of quantum trajectories in this case under an appropriate assumption.
\medskip
\begin{theorem}
Let $ \alpha $ be such that $\hat q_{\alpha}^{\hat \theta}(0) \neq 0$. Then
$$\lim\limits_{n \to \infty} \frac{1}{n}\ln \o \frac{\hat q_{\alpha}^{\hat \theta}(n)}{\hat q_{\Upsilon}^{\hat \theta}(n)} \c =S(\P_{\Upsilon}|\P_{\Upsilon}^{\hat\theta}) - S(\P_{\Upsilon}|\P_{\alpha}^{\hat\theta})~~a.s.,$$
where $\Upsilon$ is the random variable designating the pointer state selected by the trajectory $\rho_n$. Therefore, if for all $\alpha$, $q_\alpha^{\hat \theta}(0)\neq 0,$ and
\begin{equation}
\underset{\beta \in \mathcal P}{argmin~} S(\P_{\alpha}|\P_{\beta}^{\hat\theta})=\alpha,
\label{eq:argmin}
\end{equation}
then
$$\lim\limits_{n \to \infty}q_{\Upsilon}^{\hat\theta}(n)=1~~a.s,$$
equivalently
$$\lim\limits_{n \to \infty}\rho_n^{\hat\theta}=\ket{\Upsilon}\bra{\Upsilon}~~a.s.$$
\label{thm:entropy}
\end{theorem}
\begin{proof}
One can obtain a similar recurrence relation for $ \hat q_{\alpha}^{\hat\theta}(n)$ as before. This allows us to show that the following limit holds 
\begin{equation*}
\lim\limits_{n \to \infty} \frac{1}{n}\ln \o \frac{\hat q_{\alpha}^{\hat \theta}(n)}{\hat q_{\Upsilon}^{\hat \theta}(n)}\c=\sum_i  p(i|\Upsilon) \ln \o \frac{ p^{\hat\theta}(i|\alpha)}{ p^{\hat\theta}(i|\Upsilon)} \c.
\end{equation*}
The above expression can be written as follows
\begin{align*}
\lim\limits_{n \to \infty} \frac{1}{n}\ln \o \frac{\hat q_{\alpha}^{\hat \theta}(n)}{\hat q_{\Upsilon}^{\hat \theta}(n)}\c&=\sum_i  p(i|\Upsilon)\ln \o\frac{ p(i|\Upsilon)}{ p^{\hat\theta}(i|\Upsilon)}\c\\
&-\sum_i  p(i|\Upsilon)\ln \o\frac{ p(i|\Upsilon)}{ p^{\hat\theta}(i|\alpha)}\c\\
&=S(\P_{\Upsilon}|\P_{\Upsilon}^{\hat\theta}) - S(\P_{\Upsilon}|\P_{\alpha}^{\hat\theta}).
\end{align*}
Now if $\underset{\beta \in \mathcal P}{argmin~} S(\P_{\alpha}|\P_{\beta}^{\hat\theta})=\alpha$ for all $\alpha,$ then the above expression is  negative for any value $\alpha\neq\Upsilon,$ hence necessarily $q_\alpha^{\hat\theta}(n)$ tends to zero, and because these quantities sum to one, $q_\Upsilon^{\hat\theta}(n)$ converges to one.
\end{proof}
\section{Numerical study}
\label{sec:numeric}
In this section, we consider the experimental setup designed by Laboratoire Kastler-Brossel (LKB) at Ecole Normale Sup\'erieure (ENS) de Paris. Two situations are considered, first QND measurements in presence of imperfections, second we also take into account the decoherence due to the interaction of the system with the environment.  
\subsection{QND measurement: photon box example} 
We are interested in the evolution of a quantum electrodynamics cavity state considered in \cite{sayrin2011real}. The system corresponds to a quantized trapped mode inside the cavity whose state is described in the Fock basis $\ket {n}\bra {n}$ representing $n$ photons inside the cavity. The photon number states are the pointer states in this example, here we suppose that the Hilbert space is finite dimensional and we can have no more than $n^{\textrm{max}}$ photons inside the cavity. The state of the cavity is estimated through QND measurements. This is done by considering Rydberg atoms as the meter which are sent one by one inside the cavity and measured just after. The atom can be in the ground state $\ket g\bra g$ or excited state $\ket e\bra e.$ In this experiment, there are various sources of imperfection. The sample of atoms interacting with the cavity can be empty of atoms with probability $p_0,$ one atom with probability $p_1$ and two atoms with probability $p_2.$ Other sources of imperfection can be the efficiency of the detector denoted by $\epsilon_d$ which corresponds to the probability that the detector detects an atom. The final type of error corresponds to the possibility of a false detection result, we denote by $\eta_g$ (resp. $\eta_e$) the probability that the atom is detected in $g$ (resp. $e$) while the correct one is $e$ (resp. $g$). 
Formally, here we have $7$ possibilities for the detection result $i\in\{no,g,e,gg,ge,ee\}$ corresponding to have no atom, or one atom in the state $g$ or $e,$ or two atoms both in the state $g,$ one atom in $g$ and the other in $e$ or finally the possibility to have two atoms both on the state $e.$ The following expressions give the forms of the Kraus operators corresponding to each of such detection possibilities:
\begin{align*}
&V_{no}=\sqrt{p_0}I,\quad V_g=\sqrt{p_1}\cos\phi_N,\quad V_e=\sqrt{p_1}\sin\phi_N\\
&V_{gg}=\sqrt{p_2}\cos^2\phi_N,\quad V_{ge}=V_{eg}=\sqrt{p_2}\cos\phi_N\sin\phi_N,\\
&V_{ee}=\sqrt{p_2}\sin^2\phi_N,
\end{align*}
where $\phi_N=\frac{\phi_0(N+\frac{1}{2})+\phi_R}{2},$ with $\phi_0$ and $\phi_R$ corresponding to physical parameters and $N$ denoting the photon number operator.

The elements $\eta_{ij}$ of the correlation matrix for $i,j\in\{no,g,e,gg,ge,ee\}$ is determined by the  table shown in Figure \ref{fig:eta}. The system dynamics is then described by Equation \eqref{eq:dynamic}.

The values we used to simulate the true trajectory are $p_0=0.9$, $p_1=0.05$, $p_2=0.05$, $\phi_0=0.78$, $\phi_R=-0.44$. For the correlation matrix, $\epsilon_d=0.9$, $\eta_g=0.1,$ and $\eta_e=0.1$. The true initial state $\rho_0$ was chosen as a random pure state, and the estimated initial state $\hat\rho_0$ as the completely mixed state.


\vspace{-2cm}
\subsection*{Unknown initial state and parameters}
\vspace{-1cm}
Here we seek to numerically verify the asymptotic stability of trajectories whose initial states and physical parameters are unknown. For the stability purpose, we have seen that there is no real constraint on the estimated initial state (it is enough to choose it full rank). Concerning the estimated parameters, according to Theorem \ref{thm:entropy}, it is sufficient that Equation \ref{eq:argmin} holds true, that is to say that for any pointer state, the probability distribution generated with estimated parameter is closer to the distribution generated by the same pointer state with the true parameter.
For the sake of simplicity, we suppose that only the parameters $\phi_0$ and $\phi_R$ suffer from imprecision. More precisely, we compute numerically the relative entropy between $\P_\alpha$ and $\P_\beta^{\hat \theta}$ where $\hat\theta=(\hat\phi_0,\hat\phi_R)$ is an estimation of $\theta=(\phi_0,\phi_R).$ In Figure \ref{fig:entropy} we plot the parameters that verify the mentioned condition.
\begin{figure}[H]
\centering{\includegraphics[keepaspectratio=true,scale=0.5]{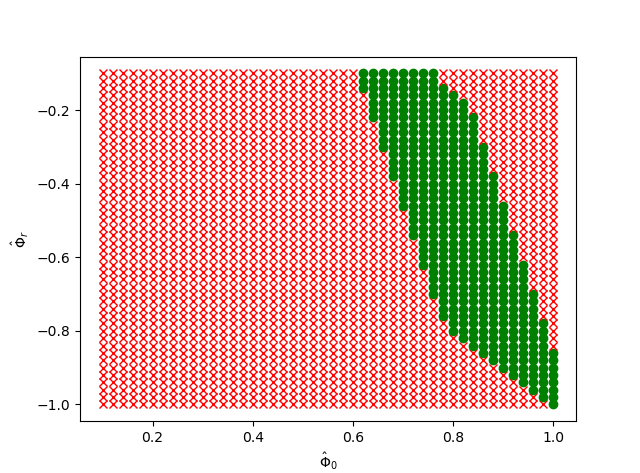}}
\caption{\footnotesize{ Validity of assumption \eqref{eq:argmin}. In green when such the assumption is verified; 
 in red when it is not. Maximum photon number is $n^{\textrm{max}}=4.$ True parameters are $\phi_0=0.78$ and $\phi_R=-0.44.$}}
\label{fig:entropy}
\end{figure}
Here is an example of the evolution of the populations in the pointer states, i.e. the quantities $q_\alpha(n)$ and $\hat q_\alpha(n)$, from $\alpha=0$ to $\alpha=4$.
\begin{figure}[H]
   \centering{\includegraphics[width=8cm,height=10cm]{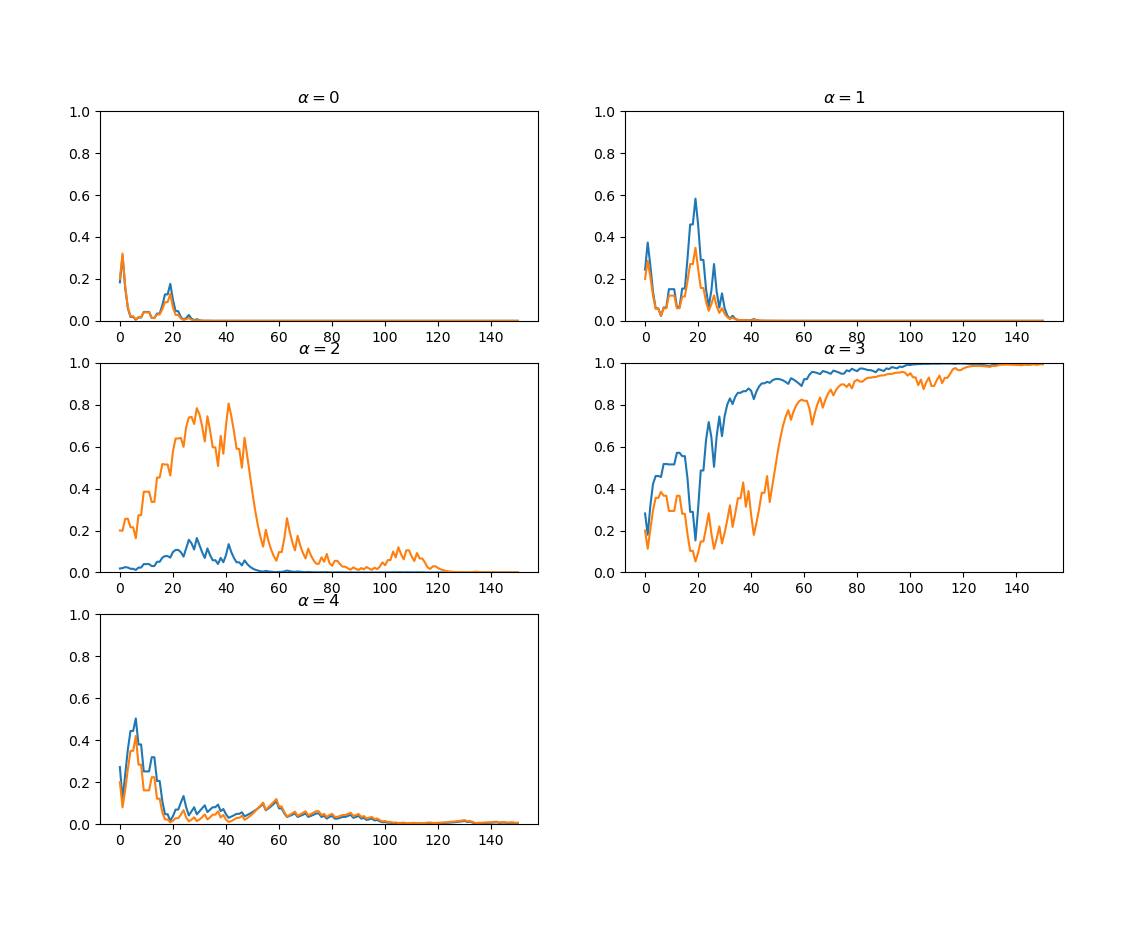}}
\caption{\footnotesize{Populations in the different pointer states over time. Blue for the true trajectory, yellow for the estimated trajectory. Estimated parameters are $\hat\phi_0=0.83$ and $\hat\phi_R=-0.40.$}}
\label{fig:population}
\end{figure}
In the following, similar to Figure \ref{fig:entropy}, we plot the parameters which verify the condition \eqref{eq:argmin} for a greater $n^{\textrm{max}},$ which shows that the region is now much narrower (the figure is enlarged for a better visibility).
\begin{figure}[ht]
\centering{\includegraphics[keepaspectratio=true,scale=0.5]{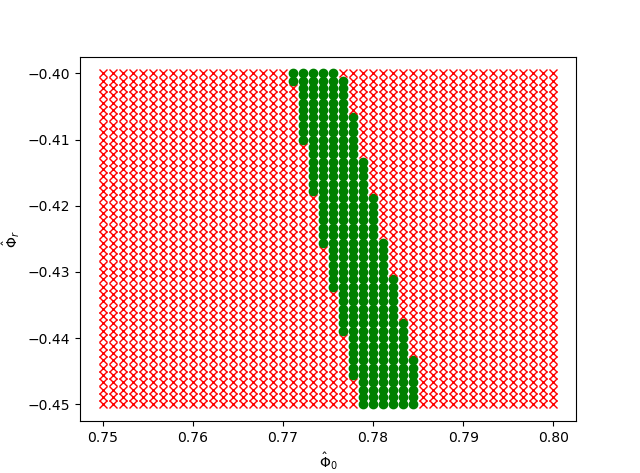}}
\caption{\footnotesize{ Validity of assumption \eqref{eq:argmin}. In green when such the assumption is verified; 
 in red when it is not. Maximum photon number is $n^{\textrm{max}}=9.$ True parameters are $\phi_0=0.78$ and $\phi_R=-0.44.$}}
\label{fig:entropy9}
\end{figure}
\subsection{Beyond the QND case: Photon box example taking into account the effect of decoherence}
In this section, we still consider the photon box example, however we now take into account the interaction between the cavity field and the environment, i.e., decoherence. The decoherence can be described by the action of the superoperator $T$ on the cavity state as follows 
$$T(\rho)=L_0\rho L_0^\dag+L_+\rho L_+^\dag+L_-\rho L_-^\dag,$$ where $L_0=I-\frac{\epsilon(1+2n_{th})}{2}N-\frac{\epsilon n_{th}}{2}I$ corresponds to the no-photon jump operator, $L_+=\sqrt{\epsilon (1+n_{th})}a$ means the capture of a photon from the environment and $L_-=\sqrt{\epsilon n_{th}}a^\dag$ represents the loss of a photon to the environment. The experimental parameters are $0\leq \epsilon, n_{th}\ll 1.$ The notations $a$ and $a^\dag$ correspond to the annihilation and creation operators respectively. 

The whole evolution can be described by $21$ Kraus operators which are in the form $L_dV_i$ with $i\in\{no,g,e,gg,ge,ee\}$ and $d\in\{0,+,-\}.$ However clearly these new Kraus operators do not satisfy the QND property. 

Here, we aim to study the asymptotic stability in presence of decoherence and imperfections through simulations.\\
This time, we cannot look at the populations in the pointer states because the Kraus operators are no longer QND. In this case, the true trajectory does not converge to a state in general. Nevertheless, we can consider the fidelity\footnote{The fidelity between two states $\rho$ and $\hat\rho$ are defined by $\mathcal{F}(\rho,\hat\rho) = \tr^2\sqrt{\sqrt{\rho}\hat\rho\sqrt{\rho}} \in [0,1].$}  between the true trajectory and the estimated one. 
 We assume that the parameters are known but we do not have access to the initial state. Figure \ref{fig:nonQND} represents the convergence of the fidelity between the true state $\rho_n$ and the estimated one $\hat\rho_n$ towards one. 
\begin{figure}[H]
   \centering{\includegraphics[scale=0.5]{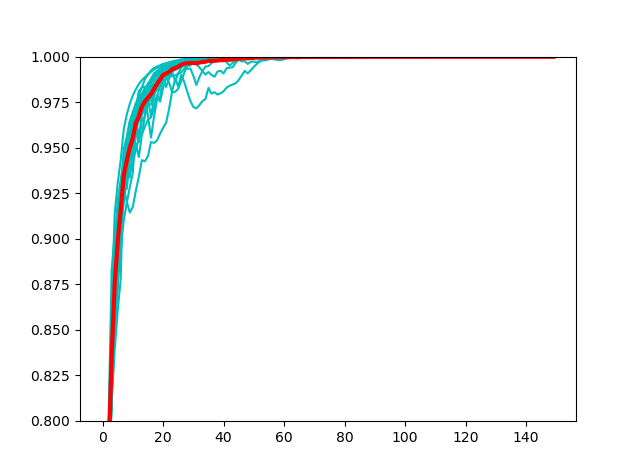}}
\caption{\footnotesize{Fidelity between the true trajectories and the estimated trajectories in presence of decoherence with unknown initial state. The red curve represents the mean value of 20 samples. Here $n^\textrm{max}=4.$}}
\label{fig:nonQND}
\end{figure}
In Figure \ref{fig:nonQNDpara},  we assume that in addition the parameters are unknown, so the estimated trajectory evolves with estimated parameters $\hat\phi_0, \hat\phi_R$. We observe that in this case the fidelity does not converge to one.  
\begin{figure}[H]
   \centering{\includegraphics[scale=0.5]{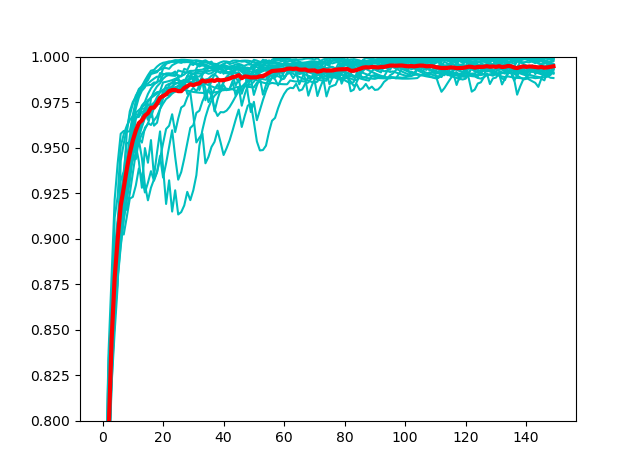}}
\caption{\footnotesize{Fidelity between the true trajectories and the estimated trajectories in presence of decoherence with unknown initial state and parameters. The red curve represents the mean value of 20 samples. Here $n^\textrm{max}=4,$ $\hat\Phi_0=0.83,$ and $\hat\Phi_R=-0.40.$}}
\label{fig:nonQNDpara}
\end{figure}
\section{Conclusion}
In this paper, we show that discrete-time quantum trajectories undergoing QND imperfect measurement are asymptotically stable with respect to initial state. This means the convergence of the estimated trajectories towards the true ones. Moreover, we consider the situation where the physical parameters are unknown and we provide a condition which ensures the asymptotic stability in this new case, not considered before in the literature as far as we know. Numerically, for the famous example of the photon box \cite{sayrin2011real}, we observe that in presence of decoherence and measurement imperfections, the estimated trajectories with arbitrary initial state converge towards the true trajectory with correct initial state even though the Kraus operators representing this situation do not satisfy QND property. It seems however that stability with respect to parameters is not ensured in general. These would be interesting questions to be further investigated.


\addtolength{\textheight}{-12cm}   



\begin{figure*}[thpb]
      \centering
\includegraphics{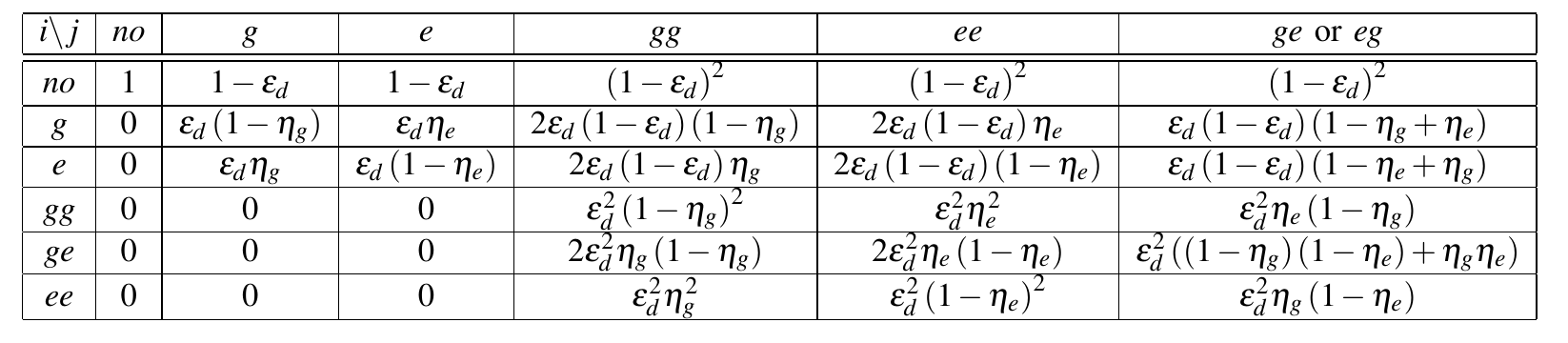}
\caption{The table showing the elements of $\eta_{ij}$, borrowed from \cite{somaraju2012design}.} 
\label{fig:eta}
\end{figure*}
\section*{Acknowledgment}
The authors thank Igor Dotsenko for interesting and motivating discussions. This work is supported by the Agence Nationale de la Recherche projects Q-COAST ANR- 19-CE48-0003, QUACO ANR-17-CE40-0007, and IGNITION ANR-21-CE47-0015. 

\bibliographystyle{unsrt}
\bibliography{refs.bib}

\end{document}